\documentclass[prl,twocolumn,amssymb,superscriptaddress,showpacs]{revtex4}
\usepackage{amsmath, verbatim, graphicx,color,amsthm, ulem, enumerate}
\usepackage{epsfig}
\usepackage{bbm}
\usepackage[latin1]{inputenc} 

\newcommand{\bra}[1]{\langle #1|}
\newcommand{\ket}[1]{|#1\rangle}
\newcommand{\ketbra}[1]{| #1\rangle \langle #1|}

\newcommand{\be}{\begin{equation}}
\newcommand{\ee}{\end{equation}}
\newcommand{\eea}{\end{eqnarray}}
\newcommand{\bea}{\begin{eqnarray}}

\newcommand{\eins}{\openone}
\newcommand{\comp}[1]{\ensuremath{\overline{#1}}}
\newcommand{\WW}{\ensuremath{W}}
\newcommand{\NN}{\ensuremath{\mathcal{N}}}

\newcommand{\BB}{\ensuremath{\mathcal{B}}}

\newcommand{\OO}{\ensuremath{\mathcal{O}}}

\newcommand{\kommentar}[1]{}
\newcommand{\trace}{{\rm Tr}}

\newcommand{\NNN}{\ensuremath{\widetilde{\mathcal{N}}}}
\newcommand{\UU}{\ensuremath{\mathcal{U}}}

\newcommand{\ie}{i.e.}

\renewcommand{\vr}{\ensuremath{\varrho}}

\newcommand{\forget}[1]{}

\newtheorem{lemma}{Lemma}

\newtheorem{proposition}[lemma]{Proposition}

\begin{document}
\title{
Taming multiparticle entanglement
}

\author{Bastian Jungnitsch}
\affiliation{Institut f\"{u}r Quantenoptik und Quanteninformation,
\"{O}sterreichische Akademie der Wissenschaften, Technikerstra{\ss}e
21A, A-6020 Innsbruck, Austria\\}
\author{Tobias Moroder}
\affiliation{Institut f\"{u}r Quantenoptik und Quanteninformation,
\"{O}sterreichische Akademie der Wissenschaften, Technikerstra{\ss}e
21A, A-6020 Innsbruck, Austria\\}
\author{Otfried G\"uhne}
\affiliation{Fachbereich Physik, Universit\"at Siegen, Walter-Flex-Stra{\ss}e 3, D-57068 Siegen, Germany\\}
\affiliation{Institut f\"{u}r Quantenoptik und Quanteninformation,
\"{O}sterreichische Akademie der Wissenschaften, Technikerstra{\ss}e
21A, A-6020 Innsbruck, Austria\\}

\date{\today}
\begin{abstract}
We present an approach to characterize genuine multiparticle entanglement 
using appropriate approximations in the space of quantum states. This leads
to a criterion for entanglement which can easily be calculated using semidefinite
programming and improves all existing approaches significantly. Experimentally, it 
can also be evaluated when only some observables are measured. Furthermore, 
it results in a computable entanglement monotone for genuine multiparticle 
entanglement. Based on this, we develop an analytical approach for the entanglement
detection in cluster states, leading to an exponential improvement compared with existing schemes.
\end{abstract}

\maketitle

{\it Introduction ---} 
The characterization of multiparticle quantum correlations is relevant 
for many physical systems like atoms in optical lattices, 
superconducting qubits or nitrogen-vacancy centers in diamond, to only name some recent
examples \cite{experiments}. In the 
field of quantum information, multiparticle entanglement is 
viewed as a resource, enabling tasks like measurement-based quantum 
computation \cite{mqc} or high-precision metrology \cite{metrology}. In spite of many efforts, the 
characterization of these correlations turns out to be difficult. Especially 
genuine multipartite entanglement, which is most important 
from the experimental point of view, remains unruly and only scattered 
results concerning its characterization are known \cite{biseparability, multient, seevinck, guehneseevinck}.

In this Letter, we derive a general method to characterize genuine multiparticle 
entanglement using suitable relaxations. This relaxed problem turns 
out to be good-natured, can be tackled with different methods and
results in a criterion that can be considered as a generalization of the 
Peres-Horodecki criterion \cite{PPTcriterion} to the multipartite case.
The goal of our work is two-fold: First, we present powerful criteria 
for genuine multiparticle entanglement, which can be efficiently evaluated 
using semidefinite programming and improve existing conditions significantly.
They work for multi-qubit systems, continuous-variable or hybrid
systems and can be evaluated, even if the mean values of only a few 
observables are known. Furthemore, they lead to a computable entanglement
monotone for genuine multiparticle entanglement.

Second, our method allows to analytically derive entanglement conditions for 
the family of cluster states \cite{clusterstate}, which are important states for tasks like 
measurement-based quantum computation. The sensitivity of these 
conditions improves exponentially with the number of qubits, which is an exponential
gain compared with the existing conditions. As a sideproduct of our investigations, 
we will also estimate the volume of the set of genuinely multipartite entangled states and 
gain insight into the geometrical form of the set of biseparable states.

{\it Situation ---} We start by considering a three-particle quantum
state $\vr.$ This state is separable with respect to some bipartition, say, $A|BC$, if it is a mixture of product states with respect to this
bipartition, \mbox{$\vr = \sum_k q_k \ketbra{\phi^{k}_{A}} \otimes \ketbra{\psi^{k}_{BC}}$},
where the $q_k$ form a probability distribution. We denote these states by $\vr_{A|BC}^{\rm sep}$. Similarly, we can define the separable states
for the two other possible bipartitions, $\vr_{B|AC}^{\rm sep}$
and $\vr_{C|AB}^{\rm sep}$.

\begin{figure}
\includegraphics[width=0.8\columnwidth]{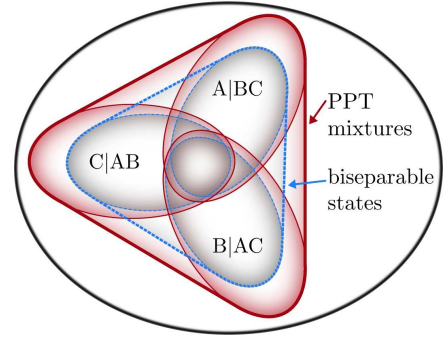}
\caption{\label{fig:PPTsets}For three qubits, there are three convex sets of states that are separable with respect to a fixed bipartition, namely the bipartitions $A|BC$, $B|AC$ and $C|AB$ (blue, dashed lines). Their convex hull (thick blue, dashed line) is the set of biseparable states. Each of the three sets is contained within the set of states that are PPT with respect to the corresponding bipartition (red, solid lines). Their convex hull forms the set of PPT mixtures (thick red, solid line).}
\end{figure}

Then, a state is called {\it biseparable} if it can be written as a mixture 
of states which are separable with respect to different bipartitions \cite{biseparability}. 
That is, one has
\begin{align}
\varrho^{\rm bs} &= p_1 \vr_{A|BC}^{\rm sep} + p_2 \vr_{B|AC}^{\rm sep} + p_3 \vr_{C|AB}^{\rm sep}.
\end{align}
On the other hand, a  state that is not biseparable is called 
{\it genuinely multipartite entangled}. Whenever we talk about multipartite entangled 
states in the following, we refer to genuinely multipartite entangled states.

To characterize multipartite entanglement, we apply the method 
illustrated by Fig.~\ref{fig:PPTsets}. Instead of states 
like $\vr_{A|BC}^{\rm sep}$ that are separable with respect to a fixed bipartition, we 
consider a larger set of states, which can be more easily characterized.
For instance, for the bipartition $A|BC$ we may consider states which have a 
positive partial transpose (PPT) \cite{PPTdef}.
It is well-known that separable states are 
also PPT \cite{PPTcriterion}. We denote such states by $\vr_{A|BC}^{\rm ppt}$ 
(and analogously for the other bipartitions).

Thus, we ask whether a state can be written as 
\begin{align}
\varrho^{\rm pmix} &= p_1 \vr_{A|BC}^{\rm ppt} + p_2 \vr_{B|AC}^{\rm ppt} + p_3 \vr_{C|AB}^{\rm ppt}.
\end{align}
We call states of this form {\it PPT mixtures}. Clearly, any
biseparable state is a PPT mixture, so proving that a state is no PPT mixture implies 
genuine multipartite entanglement. There are examples of states, which are 
PPT with respect to any bipartition, but nevertheless multipartite entangled \cite{pianimora}. Hence, 
not all multipartite entangled states can be detected in this way, but, as we will see, 
often the set of PPT mixtures is a very good approximation to the set of biseparable states.
Finally, note that all definitions can be extended to $N$ particles. Also, one may use other relaxations 
of bipartite separability, e.g. apply the criterion of Doherty et al. \cite{doherty}.

The advantage of considering PPT mixtures instead of biseparable states is 
that the set of PPT mixtures can be fully characterized with the method 
of linear semidefinite programming (SDP) \cite{sdp} --- a standard problem of constrained 
convex optimization theory. Moreover, PPT mixtures can also be characterized 
analytically.

{\it Characterization via entanglement witnesses ---}
An {\it entanglement witness} is an observable $\WW$
that is non-negative on all biseparable states, but has a negative expectation 
value on at least one entangled state. Let us first consider two particles, 
A and B. Then a {\it decomposable witness} is a witness $\WW$ that can be 
written as $\WW = P + Q^{T_A}$, where $P$ and $Q$ have no negative eigenvalues, 
(they are positive semidefinite, $P, Q \geq 0$) and $T_A$ is the partial transpose 
with respect to A \cite{decwit}. 

For more than two particles, we call a witness $\WW$ {\it fully decomposable} 
if, for every subset $M$ of all systems, it is decomposable with respect to the
bipartition given by $M$ and its complement $\overline{M}.$
This means, there exist positive semidefinite operators $P_M$ and $Q_M$, such that 
\be
\label{eq:def_decW}
{\rm for \: all} \: M: \: \WW = P_M + Q_M^{T_M} \: .
\ee
This observable is positive on all PPT mixtures, as
it is non-negative on all states which are PPT with respect to some bipartition.
But also the converse holds:

\textbf{Observation.} 
If $\varrho$ is not a PPT mixture, then there exists a fully decomposable witness $\WW$ 
that detects $\varrho$.

{\it Proof ---} 
The set of PPT mixtures is convex and compact. Therefore, for any state outside of it, 
there is a witness that is positive on all PPT mixtures. Furthermore, positivity on 
all states that are PPT with respect to a fixed (but arbitrary) bipartition implies that the witness is decomposable with respect to 
this fixed (but arbitrary) bipartition \cite{decwit}. Thus, $\WW = P_M + Q_M^{T_M}$ for any $M$. \hfill $\square$

{\it Practical evaluation ---} 
To find a fully decomposable witness for a given state, the convex optimization technique SDP becomes important, since it allows us to optimize over all fully decomposable witnesses.
Given a multipartite state $\varrho$, the search is given by
\begin{align}
\label{eq:sdp1}
&\min \:\trace  (\WW \varrho)\\
&\begin{aligned}
{\mbox{s.t.}} \: &\trace(\WW) = 1 \: \mbox{and for all} \: M : \nonumber \\
&\WW = P_M + Q_M^{T_M}, Q_M \geq 0 ,\:P_M \geq 0 \nonumber \:.
\end{aligned}
\end{align}
The free parameters are given by $W$ 
and the operators $P_M$ for every subset $M$. If the minimum in Eq.~(\ref{eq:sdp1}) is negative, 
$\varrho$ is not a PPT mixture and hence genuinely multipartite  
entangled. The operator $\WW$ for which the negative minimum is 
obtained is a fully decomposable witness. 
Note that, due to \mbox{$X^{T_M} = (X^T)^{T_{\overline{M}}}$} and 
$X \geq 0 \Leftrightarrow X^T \geq 0$, a witness that is decomposable 
with respect to $M$ is also decomposable with respect to $\overline{M}$. 
Thus, one needs to check only half of all subsets in practice.

Eq.~(\ref{eq:sdp1}) has the form of a semidefinite 
program \cite{sdp}. In contrast to usual optimization problems, global optimality of an SDP can be certified 
and the solution can efficiently be computed via interior-point methods. 
In practice, Eq.~(\ref{eq:sdp1}) can be solved with few lines of code, using e.g. the parser YALMIP \cite{yalmip} and, as solvers, SeDuMi \cite{sedumi} or SDPT3 \cite{sdpt3}. Our implementation in MATLAB named PPTMixer can be found online \cite{matlabcentral}. 

Let us discuss two variations of Eq.~(\ref{eq:sdp1}).
First, in order to  reduce the number of free parameters, one can restrict 
oneself to witnesses $\WW$ that obey $\WW^{T_M} \geq  0$ for all $M$, \ie, 
$P_M = 0 $ for all $ M$. In the following, we will call these witnesses {\it fully PPT witnesses}. 
For bipartite systems, decomposable witnesses and fully PPT witnesses detect
the same states. For the multipartite case, fully PPT 
witnesses are not as good as fully decomposable witnesses, 
but they are easier to characterize.

Second, this SDP can also be modified to account for the case that, instead 
of a full tomography, only a restricted set of observables is measured.
Let $\OO = \lbrace O_1,... O_k \rbrace$ be such a set of observables. Then, 
adding $\WW = \sum_{i=1}^{k} \lambda_i O_i$ to the constraints in Eq.~(\ref{eq:sdp1}) 
results in an SDP that searches for witnesses which can be evaluated knowing these observables. Note that
for this program the free parameters are given by the real numbers $\lambda_i$ 
and their number might be considerably smaller than in Eq.~(\ref{eq:sdp1}). 
If the minimum in Eq.~(\ref{eq:sdp1}) is non-negative, there exists a
PPT mixture with expectation values $\langle O_i \rangle$. However, one may then
add further observables to $\OO$ and run the SDP again. Repeating this procedure gives 
more and more sensitive tests. We will discuss an example later.
In practice, this program can even decide separability if the $O_i$ 
have already been measured, so it can be used to gain new insights into already performed
experiments.

But before proceeding to the examples, let us note three more facts.
First, in the formulation no dimension of the Hilbert space is fixed. 
Consequently, our approach is valid for any dimension and combined with the methods of Ref.~\cite{schukinvogel} it can be directly used to study multipartite entanglement in continuous-variable or hybrid systems \cite{haeseler}. For continuous variables, it can be employed complementary to the methods of Ref.~\cite{hylluseisert}.

Second, our approach can also be used to {\it quantify}
genuine multipartite entanglement. If in Eq.~(\ref{eq:sdp1})
the trace normalization $\trace(\WW)=1$ is replaced by 
$0 \leq P_M \leq \openone$ and $0 \leq Q_M \leq \openone$, the 
negative witness expectation value is a multipartite entanglement 
monotone, since it obeys the following properties: (i) It vanishes on all biseparable states. (ii) It is convex. (iii) The quantity does not increase under protocols that consist of local operations of each party and classical communication between them. (iv) It is invariant under local basis changes. While most of these properties are straightforward to see --- in particular, (iv) is implied by (iii) ---, the proof of property (iii) is more technical (cf. Appendix A). Note that, in the bipartite case, this monotone becomes the negativity~\cite{negativity}.

Third, as mentioned before, there are other possible choices of supersets for the set of separable states, e.g. the set of states that have a symmetric extension on a larger Hilbert space (cf. Appendix B and \cite{doherty}).

{\it Numerical examples ---} 
We test the criterion of Eq.~(\ref{eq:sdp1}) for important pure three- and four-qubit 
states prepared in many experiments \cite{states}, using the white noise tolerance
as a figure of merit. It is given by the maximal amount $p_{\rm tol}$ of white noise for which the state 
\mbox{$\varrho(p_{\rm tol}) = (1-p_{\rm tol})\ketbra{\psi} +{p_{\rm tol}} \eins/ {2^n}$}
is still detected as entangled \cite{whitetol}.
Table~\ref{tab:noiserob} shows the white noise tolerances of our criterion, 
compared with the most robust criteria so far.

\begin{table}[t]
\begin{tabular}{|c|c|c|c|}
\hline
state & \multicolumn{2}{c|}{white noise tolerances $p_{\rm tol}$} \\
\hline
& fully decomposable & before\\
\hline
$\ket{GHZ_3}^{\star}$ & $0.571$ & $0.571$ \cite{guehneseevinck} \\ 
$\ket{GHZ_4}^{\star}$ & $0.533$ & $0.533$ \cite{guehneseevinck} \\ 
$\ket{W_3}^{\star}$& $0.521$ & $0.421$ \cite{guehneseevinck}\\ 
$\ket{W_4}$& $0.526$ & $0.444$ \cite{guehneseevinck}\\ 
$\ket{Cl_4}^{\star}$& $0.615$ & $0.533$ \cite{tokunaga}\\ 
$\ket{D_{2,4}}$ & $0.539$ & $0.471$ \cite{huberdicke}\\ 
$\ket{\Psi_{S,4}}$& $0.553$ & $0.317$ \cite{toolboxpaper}\\ 
\hline
\end{tabular}
\caption{\label{tab:noiserob}
White noise tolerances of the fully decomposable witnesses obtained by the SDP 
of Eq.~(\ref{eq:sdp1}) compared with the corresponding tolerances of the most robust criteria 
known so far. For the states marked by $^{\star}$, we verified that adding more 
white noise than what is tolerated by Eq.~(\ref{eq:sdp1}) results in a biseparable 
state, so the values are optimal.}
\end{table}

Strikingly, the tolerances of the witnesses obtained by our SDP 
are significantly higher than previous ones. For the GHZ and the W state of three 
qubits and the GHZ and the linear cluster state of four qubits, we even obtain the best white 
noise tolerance possible, since we are able to show that for a larger amount of 
white noise the state becomes biseparable (see Appendix C). This shows that our 
criterion is indeed optimal for these cases.

To show that the criterion of Eq.~(\ref{eq:sdp1}) works well for a restricted 
set of observables, we consider the four-qubit Dicke state with two excitations 
$\ket{D_{2,4}}$ \cite{states}. For this state, the SDP yields a witness $W_{D}$ 
(see Appendix D) that consists of the observables $\OO = \lbrace X^{\otimes 4},  
Y^{\otimes 4},  Z^{\otimes 4},  X_1 X_2 Y_3 Y_4, X_1 X_2 Z_3 Z_4, Y_1 Y_2 Z_3 Z_4\rbrace$, 
their distinct permutations and other observables that can be measured in the same run. 
For example, a local measurement of $X_1 X_2 X_3 X_4$ yields knowledge of the expectation value
of $X_1 X_2 \eins_3 X_4$. The SDP finds a witness consisting of $O_1 =X^{\otimes 4}$, $O_2 =Y^{\otimes 4}$ and observables obtained by replacing some Paulis by the identity. Already with these observables, the  white noise tolerance is  $p_{\rm tol}^{(2)} \approx 0.29495$. We can proceed 
in this way and use additional observables $O_i$ from the set $\OO$ --- including their permutations and observables obtained by replacing Pauli operators by $\eins$--- to produce strictly stronger witnesses $W_{D}^{(i)}$. Their white noise tolerances $p_{\rm tol}^{(i)}$ are $p_{\rm tol}^{(3)} \approx 0.38379$, 
$p_{\rm tol}^{(4)} \approx 0.38383$, $p_{\rm tol}^{(5)} \approx 0.45200$ and finally $p_{\rm tol}^{(6)} \approx 0.53914$ as in Table \ref{tab:noiserob}, since $W_{D} = W_{D}^{(6)}$.

Third, we compute a lower bound
on the volume of genuinely multipartite entangled states. We created samples 
of $10^4$ random mixed three-qubit states uniformly distributed 
with respect to the Hilbert-Schmidt distance (or the Bures distance) and check whether they 
are genuinely multipartite entangled. 
6.28 \%  (Bures: 10.32 \%)
were detected by fully decomposable and 0.44 \%  (Bures: 1.06 \%) by fully PPT 
witnesses. As expected, fully PPT witnesses detect fewer states. 

While the problem can still be tackled numerically for six or seven qubits, in recent experiments up to 14 ions have been coherently manipulated \cite{blatt}. Therefore, we study analytical witnesses which can be generalized to an arbitrary number of qubits in the following.

{\it Analytical results ---} A fully decomposable witness for the four-qubit linear cluster state $\ket{Cl_4}$ \cite{states} that is obtained by the SDP of Eq.~(\ref{eq:sdp1}) is given by 
\begin{align}
\label{eq:cl4witness}
W_{\rm Cl4} = \frac{1}{2}\eins - \ketbra{Cl_4} - \frac{1}{8} \left(\eins-g_1\right) \left(\eins-g_4\right)\:,
\end{align}
where $g_1 = Z_1 Z_2 \eins_3 \eins_4$ and $g_4 = \eins_1 \eins_2 Z_3 Z_4$ are two of the generators of the cluster state's so-called stabilizer group. This witness detects more states than the usual projector witness $W_{\rm proj} = \frac{1}{2}\eins - \ketbra{Cl_n}$, since $W_{\rm Cl4}$ is obtained from $W_{\rm proj}$ by subtracting a positive operator $P_{+}$. For $n$ qubits, the generators are, after a local basis change, $g_1 = X_1 Z_2$, $g_i= Z_{i-1} X_i Z_{i+1}$ for $1 < i < n$ and $g_n = Z_{n-1} X_n$. Then, the n-qubit linear cluster state is defined by $\ketbra{Cl_n} = 2^{-n} \prod_{i=1}^{n} \left( \eins+g_i \right)$. The construction of the four-qubit cluster state witness can be generalized to an arbitrary number of qubits (see Appendix E). For seven qubits, e.g., a witness is given by
\vskip -8 pt

\begin{align}
\label{eq:cl7witness}
W_{\rm Cl7} = \: &\frac{1}{2}\eins - \ketbra{Cl_7} - \frac{1}{16} \left[\left(\eins-g_1\right)\left(\eins-g_4\right)\left(\eins-g_7\right) \right. \nonumber \\
& + \left(\eins+g_1\right)\left(\eins-g_4\right)\left(\eins-g_7\right) \nonumber \\
& +\left(\eins-g_1\right)\left(\eins+g_4\right)\left(\eins-g_7\right) \nonumber \\
& \left. +\left(\eins-g_1\right)\left(\eins-g_4\right)\left(\eins+g_7\right)\right] \:.
\end{align}
For the case of $n$ qubits, the white noise tolerance is 
\be
\label{eq:convergence}
p_{\rm tol} = \left(1-2^{-n+1}+(k+1) 2^{-k}\right)^{-1} \xrightarrow{n \rightarrow \infty} 1
\ee
where $k = \lfloor \frac{n+2}{3} \rfloor$. This result is remarkable for several reasons. First, $W_{\rm Cln}$ is the first example of a witness for genuine multipartite entanglement so far whose white noise tolerance converges to one for an increasing number of qubits. Thus, the volume of the largest ball inside the biseparable states around the totally mixed state approaches zero. 
A similar scaling behavior of the entanglement in the cluster state has been found in Ref. \cite{clusterscale}. Note that, however, they considered full separability and not genuine multipartite entanglement. For full separability, this scaling behavior is not surprising, since it is known that the largest ball of fully separable states around the totally mixed states shrinks with increasing particle number \cite{fullsepball}. Moreover, the white noise tolerance of Eq.~(\ref{eq:convergence}) corresponds to a required fidelity $F_{\rm req} \approx 1- p_{\rm tol} \approx k 2^{-k}$ for large $n$ and therefore decreases exponentially fast with growing $n$. In contrast, the fidelity needed to detect entanglement using $W_{\rm proj}$ equals one half, independent of the particle number. Interestingly, this exponential improvement comes with very low experimental costs, since the additional term $P_{+}$ can be measured with only one experimental setting. Finally, note that our construction induces a similar construction for the 2D cluster state.

{\it Discussion ---} In this Letter, we presented an easily implementable criterion for genuine multipartite entanglement. We demonstrated its high robustness, connected it to entanglement measures and provided powerful witnesses for an arbitrary number of qubits.

Due to its versatility, the criterion can be used to characterize the entanglement in various physical systems, e.g. in ground states of spin models undergoing a quantum phase transition. Moreover, it is a promising tool to study multipartite entanglement in continuous-variable systems. Finally, we believe that, as an easy-to-use scheme, it will be valuable for the analysis of experimental data that do not constitute a whole tomography.

We thank M. Kleinmann, T. Monz, S. Niekamp, A.~Osterloh, M. Piani and G. T\'oth for discussions and acknowledge support 
by the FWF (START Prize and SFB FOQUS).

\bibliographystyle{apsrev}

\section{Appendix}

This appendix consists of five sections. In Appendix A, we prove the properties (i) to (iv) of the entanglement monotone given in the main text. Appendix B gives an example for another superset that can be used to approximate the set of states which are separable with respect to a fixed bipartition. In the third section of the Appendix, we prove the optimality of some white noise tolerances given in Table I in the main text. Afterwards, we write down a fully decomposable witness for the Dicke state with two excitations in Appendix D. In the main text, this witness was used to illustrate the case in which one considers only witnesses that consist of a restricted set of observables. Finally, in Appendix E we provide a general theory for witnesses of linear cluster states. Examples for these witnesses have been given in the main text.

\subsection{Appendix A}
In this appendix, we introduce an entanglement monotone for genuine multipartite entanglement that is motivated by the notion of fully decomposable entanglement witnesses. Let us point out that this defined quantity equals the negativity for the bipartite case \cite{negativityapp}, hence it can be considered as an extension of the negativity to the multipartite case. 

\begin{proposition} For a generic multipartite state $\varrho$, consider 
\begin{align}
&N(\varrho)= -\min_{W \in \mathcal{W}} \trace(\varrho W), \\
&\mathcal{W}= \left\{ W \big| \; \forall M: \exists\; 0 \leq P_M, Q_M \leq \mathbbm{1}: W=P_M + Q_M^{T_M} \right\}\!, \nonumber
\end{align}
where $M$ stands for a possible partition of the subsystems. Then $N(\varrho)$ has the following properties:
\begin{itemize}
\item $N(\varrho^{\rm bs})=0$ for all biseparable states $\varrho^{\rm bs}$.
\item $N[\Lambda_{\rm LOCC}(\varrho)]\leq N(\varrho)$ for all full LOCC operations.
\item $N(U_{\rm loc} \varrho U^\dag_{\rm loc})=N(\varrho)$ for local basis changes $U_{\rm loc}$.
\item $N(\sum_i p_i \varrho_i) \leq \sum_i p_i N(\varrho_i)$ holds for all convex combinations $\sum_i p_i \varrho_i$.
\end{itemize}
\end{proposition}

\begin{proof}
The first statement follows directly from the fact that any fully decomposable witness can only detect genuine multipartite entanglement, hence the expectation value satisfies $\trace(\varrho^{\rm bs} W) \geq 0$ and vanishes for $W=0$. 

For the second property we effectively show $N[\Lambda(\varrho)]\leq N(\varrho)$ for all trace-preserving, completely positive operations $\Lambda(\varrho)=\sum_i K_i \varrho K^\dag_i$ with $\sum_i K^\dag_i K_i = \mathbbm{1}$ that admit a fully separable operator-sum representation, which means that each operator $K_i=A_i \otimes B_i \otimes \dots \otimes F_i$ has a tensor product form. This set of operations defines a superset to the set of full LOCC operations, so we verify the above property for an even larger set of possible operations \cite{bennett}. Let us point out that for each operation $\Lambda$, there exists an adjoint operation $\Lambda^\dag(Y)=\sum_i K_i^\dag Y K_i$, that satisfies the identity $\trace[\Lambda(X) Y]=\trace[X \Lambda^\dag(Y)]$ for all linear operators $X,Y$. The trace-preserving condition for $\Lambda$ translates to a unital condition for the adjoint map $\Lambda^\dag(\mathbbm{1})=\mathbbm{1}$. 

Let us first prove the following statement: For each trace-preserving, separable operation $\Lambda$ and for any decomposable operator $W$ the observable $\Lambda^\dag(W)$ is decomposable as well. Suppose that $W=P + Q^{T_M}$ is an appropriate decomposition~\cite{subscripts} with respect to a chosen partition $M$. Because of linearity we obtain $\Lambda^\dag(W)=\Lambda^\dag(P) + \Lambda^\dag(Q^{T_M})$. First, we want to check ``normalization'' $0\leq \Lambda^\dag(P) \leq \mathbbm{1}$. Complete positivity provides $\Lambda^\dag(P)\geq~0$ since $P\geq~0$ is positive semidefinite itself. If one applies the adjoint map to $(\mathbbm{1}-P)\geq~0$ and employs the unital condition one obtains 
\begin{equation}
\Lambda^\dag(\mathbbm{1}-P)=\Lambda^\dag(\mathbbm{1})-\Lambda^\dag(P)=\mathbbm{1}-\Lambda^\dag(P)\geq 0,
\end{equation}
hence the upper bound $\Lambda^\dag(P)\leq~\mathbbm{1}$ holds as well. We can apply the same argument to $\Lambda^\dag(Q)$ if we can fulfil the identity $\Lambda^\dag(Q^{T_M})=[\tilde\Lambda^\dag(Q)]^{T_M}$ with $\tilde{\Lambda}$ being completely positive and unital. 
Using the assumed tensor product structure of each operator $K_i$ it is straightforward to deduce the Kraus operators of the liner map $\tilde \Lambda$ satisfying this identity. These operators $\tilde K_i = \tilde A_i \otimes \tilde B_i \otimes \dots \otimes \tilde F_i$ are given by $\tilde A_i = A_i$ if $A \not \in M$ and $\tilde A_i = A^*_i$ if $A \in M$, and similar relations for all other operators. Let us point out that this is the only step where one explicitly needs the separable operator structure. 

Via this statement one finally obtains     
\begin{eqnarray}
N[\Lambda(\varrho)] &=& - \min_{W \in \mathcal{W}} \trace[\Lambda(\varrho) W] = - \min_{W \in \mathcal{W}} \trace[\varrho \Lambda^\dag(W)] \nonumber \\ 
&\leq& - \min_{W \in \mathcal{W}} \trace[\varrho  W] = N[\varrho].
\end{eqnarray}
The inequality holds since $\Lambda^\dag(W)$ is decomposable again, hence the optimization over the complete set of decomposable entanglement witnesses can only produce a lower (negative) expectation value. 

Invariance under local basis changes is a direct consequence of the previous property. Since any local basis change $U_{\rm loc}$ represents an invertible full LOCC operation, one obtains 
\begin{equation}
\nonumber
N(\varrho) \geq N(U_{\rm loc} \varrho U^\dag_{\rm loc}) \geq N[U_{\rm loc}^\dag (U_{\rm loc} \varrho U^\dag_{\rm loc}) U_{\rm loc}]= N(\varrho).
\end{equation}
Thus a local basis does not change the value of $N(\varrho)$.

The convexity statement 
\begin{align}
N\left(\sum_i p_i \varrho_i\right)= &-\min_{W} \sum_i p_i \trace(\varrho_i W) \nonumber \\
\leq &\sum_i p_i \left[ - \min_W \trace(\varrho_i W) \right] \nonumber \\
= & \sum_i p_i N(\varrho_i),
\end{align}
follows from linearity of the trace and the fact that if one performs independent optimizations then the obtained expectation value can only be smaller. This completes the proof.
\end{proof}

\subsection{Appendix B}

Here, we would like to give another example of a superset that one can
use to approximate the separable states of a given bipartition. 
For the bipartition $A|BC$, one can consider states that
 possess a symmetric extension to $k$~copies of system $A$, 
cf.~Refs.~\cite{dohertyapp}. This means that the given state 
$\varrho_{A|BC}$ can be written as the reduced state of a multipartite 
state $\varrho_{A_1\dots A_k|BC}$ that is invariant under all possible 
permutations of the copied subsystems.
Every separable state necessarily satisfies this 
extension condition for any number of copies and we denote states 
of this class by $\varrho_{A|BC}^{\rm sym_k}$. 
Consequently, we ask whether a three-particle state can be decomposed as
\begin{align}
\varrho &= p_1 \vr_{A|BC}^{\rm sym_k} + p_2 \vr_{B|AC}^{\rm sym_k} + p_3 \vr_{C|AB}^{\rm sym_k}.
\end{align}
Any biseparable state can be written in this way, hence if this expansion fails, genuine multipartite entanglement is detected. This approach has appealing properties: With increasing number of copies, these supersets converge to the set of separable states \cite{dohertyapp}. Moreover, it is again possible to characterize such decompositions using entanglement witnesses that can be tackled via SDP. These witnesses are such that for all possible bipartitions $M$, the operator $\WW$  must be a bipartite entanglement witness for the case of $k$ symmetric extensions as given in Ref.~\cite{dohertyapp}.

\subsection{Appendix C}

In this appendix, we prove that our criterion is optimal for three-qubit GHZ and W states
and for four-qubit GHZ and cluster states (cf. marked states of Table I).

First, for the GHZ states mixed with white noise, the noise thresholds at which the states 
become biseparable were already derived in Ref.~\cite{guehneseevinckapp} and the values 
of $p_{\rm tol}$ for our criterion coincide with these values.

Second, for the three-qubit W state $\ket{W_3}$ we introduce 
a variation of our relaxation idea, in order to show that for 
a white noise contribution of $p \approx 0.521$ the corresponding 
state is biseparable.

Consider the tripartite case. Instead of supersets to approximate 
separable states for bipartition $A|BC$, we now consider a strictly 
smaller set of states, \ie, any state of this subset must necessarily 
be separable. States of this restricted class are denoted as 
$\tilde \varrho_{A|BC}$, and similar for other bipartitions. 
If a given state has a decomposition  
\begin{align}
\label{eq:biseprelax}
\varrho &= p_1 \tilde \vr_{A|BC} + p_2 \tilde \vr_{B|AC} + p_3 \tilde \vr_{C|AB},
\end{align}
\noindent then it is clearly biseparable. The main difficulty 
is to find operational subset approximations, though results 
have been obtained along this direction recently~\cite{spedalieri,navascues09,thiang10}. 
For the three-qubit case we employ an approximation that 
exploits the low dimensionality of the problem. Instead 
of all separable states of a $2 \otimes 4$ system we 
consider only separable states $\tilde \varrho_{A|BC}$ 
that are supported on the symmetric subspace of system 
$BC$. Since $\text{Sym}(\mathcal{H}_{BC}) \simeq \mathbbm{C}^3$, 
one effectively has a $2 \otimes 3$ system for which the PPT 
condition $\tilde \varrho_{A|BC}^{T_A} \geq 0$ becomes necessary and sufficient for separability~\cite{entanglement_witness}. Similar definitions are 
used for the other bipartitions. The search for an explicit 
decomposition as given by Eq.~(\ref{eq:biseprelax}) can be 
cast into the form of a SDP again. 

The solution for $\varrho(p)=(1-p)\ket{W_3}\bra{W_3}+p\mathbbm{1}/8$ is as follows: The separable states for $B|AC$ and $C|AB$ are the same as for the bipartition $A|BC$ but with appropriate permutations of the subsystems. All probabilities equal $p_i=1/3$. Thus, it is left to characterize $\tilde \vr_{A|BC}$. The corresponding eigenspectrum is given by $\lambda_1=\lambda_2=p/8$, $\lambda_3=3p/8$ and $\lambda_4=1-5p/8$ with respect to the eigenstates $\ket{\chi_1}=\ket{000}$, $\ket{\chi_2}=\ket{111}$,
\begin{eqnarray}
\ket{\chi_3}&=&\left( 2 \sqrt{2} \ket{1}\otimes \ket{\psi^+} - \ket{011} \right)/3,\\
\ket{\chi_4}&=&a \ket{0}\otimes \ket{\psi^+} + \sqrt{1-a^2} \ket{100},
\end{eqnarray}
with the abbreviation $\ket{\psi^+}=(\ket{01}+\ket{10})/\sqrt{2}$. The parameter $a \in \mathbbm{R}$ is determined by the linear condition given by Eq.~(\ref{eq:biseprelax}) and will be the root of a non-linear function. The only remaining condition is $\tilde \varrho_{A|BC}^{T_A} \geq 0$ which is satisfied if $p \geq \bar p$ with a critical white noise level given by
\begin{equation}
\bar p =\frac{1}{382}\left[367-71\sqrt{3}-\sqrt{2894\sqrt{3}-2988}\right] \approx 0.52102.
\end{equation} 
At this white noise level one has $a \approx 0.98716$. The noise level $\bar p$ coincides with 
the maximal $p_{\rm tol}$ for our criterion.

Finally, for the four-qubit cluster state $\ket{Cl_4}$ we used the algorithm of 
Ref.~\cite{barreiro}, that can be used to prove different forms of multipartite 
separability. We found that a state with a noise contribution of $p=0.616$ is biseparable, 
while the witness $W_{\rm Cl4}$ of Eq.~(6) in the main text detects it for $p<8/13 \approx 0.6154$ as genuine 
multipartite entangled.

\subsection{Appendix D}
The fully decomposable witness for the Dicke state $\ket{D_{2,4}}$ which was used in the main text to demonstrate entanglement detection by measuring a restricted set of observables is given by
\begin{widetext}
\begin{align}
W_{D}\! = \! \: \frac{1}{16} & \left[ \eins + \alpha_1 X^{\otimes 4} + \alpha_1 Y^{\otimes 4} + \alpha_2 Z^{\otimes 4} + \alpha_3 \left( X_1 X_2 Y_3 Y_4 + {\rm perms} \right)+\alpha_4 \left( Z_1 Z_2 Y_3 Y_4 + {\rm perms} \right) \right.  \\
& \left. + \alpha_4 \left( Z_1 Z_2 X_3 X_4 + {\rm perms} \right) \!+\alpha_5 \left( X_1 X_2 \eins_3 \eins_4 + {\rm perms} \right)\! +\alpha_5 \left( Y_1 Y_2 \eins_3 \eins_4 + {\rm perms} \right)\!+\alpha_6 \left( Z_1 Z_2 \eins_3 \eins_4 + {\rm perms} \right)\right]\: \!. \nonumber
\end{align}
Here, $X_1 X_2 Y_3 Y_4 + {\rm perms}$ is the sum over all distinct permutations of $X_1 X_2 Y_3 Y_4$. Moreover, $\alpha_1 = 0.014,$ $\alpha_2  = -0.095,$ $\alpha_3  = 0.0046,$ $\alpha_4  = 0.16,$ $\alpha_5 = -0.14,$ $\alpha_6 = -0.15$.
\end{widetext}

\subsection{Appendix E}
Here, we prove the generalized n-qubit version $W_{\rm Cln}$ of the cluster state witnesses given in the main text to be fully decomposable witnesses. For four and seven qubits, $W_{\rm Cln}$ reduces to forms of Eqs. (6) and (8) of the main text, respectively. First, let us briefly introduce the notation that we use.\\
Consider an arbitrary graph $G = (V,E)$ that is defined by a set $V$ of vertices which correspond to qubits and a set $E$ of edges that connect some of these vertices. \\
Then, by $\NN(i)$ we denote the {\it neighborhood of qubit $i$}, \ie, the set of all vertices that are connected to $i$ by an edge. Moreover, we define $\NNN(i) = \NN(i) \cup \lbrace i \rbrace$.\\
The commuting operators defined by
\be
\label{eq:generators}
g_i = X_i \prod_{k \in \NN(i)} Z_k, \: i = 1, \dots ,n
\ee
are the generators of the so-called {\it stabilizer group}, \ie, $\lbrace S_1 , \dots , S_{2^n} \rbrace = \langle g_1, \dots, g_n\rangle$. In other words, every stabilizer $S_i$ can be written as a product of some (or all) of these generators $g_i$.\\
The {\it graph state $\ket{G}$} associated to the given graph $G$ is then uniquely defined by
\be
g_i \ket{G} = \ket{G}, \; \forall \: i = 1, \dots ,n \:.
\ee
One can also associate an orthonormal basis to a graph $G$. The elements $\ket{a_1 \dots a_n}_G$ of this so-called {\it graph state basis} are defined by
\be
g_i \ket{a_1 \dots a_n}_G = (-1)^{a_i} \ket{a_1 \dots a_n}_G, \; \forall \: i = 1, \dots ,n \:.
\ee
Consequently, $\ket{G} = \ket{0 \dots 0}_G$. Also, note that projectors on these vectors can be written as
\be
\label{eq:projectorform}
_G \ketbra{a_1 \dots a_n}_G = \prod \limits_{i = 1}^{n} \frac{(-1)^{a_i} g_i+\eins}{2} \:.
\ee
For better readability, we will drop the subscript $G$ in the following and write $\ket{a_1 \dots a_n}_G = \ket{\vec{a}}$. Thus, except for a single exception that we will explicitly mention, all vectors in the following have to be understood in the graph state basis. However, it is important to note that all partial transpositions will be done w.r.t the computational basis.\\
A {\it partition} $M$ is a subset of the set of all qubits, while we refer to the set made up of partition $M$ and its complement $\comp{M}$ as {\it bipartition} $M\vert \comp{M}$.\\
With the notation clarified, we first prove three lemmas to prepare the proof that the operator $W_{Cl}$ presented in the paper is a fully decomposable witness. The first of these lemmas starts from two orthogonal graph state basis vectors and shows which kind of partial transpositions one can apply to one of these vectors such that the two resulting vectors are orthogonal. Lemma \ref{lem:rule2} supplies an upper bound on the eigenvalues of any partially transposed state in terms of the state's Schmidt coefficients. Finally, Lemma \ref{lem:rule3} provides some kind of construction manual. Given a qubit and starting from a graph state basis vector, it tells us how we can construct expressions out of this basis vector which are invariant under a partial transposition on the given qubit.\\
These lemmas will then be used in Proposition \ref{lem:witness} to prove that the presented operator $W_{\rm Cln}$ is a fully decomposable witness.

\begin{lemma}
Given a graph $G = (V,E)$ of $n$ qubits and an arbitrary bipartition $M\vert \comp{M}$ of these qubits. Let $\ket{\vec{a}}$ and $\ket{\vec{b}}$ be two arbitrary states in the associated graph state basis. Let $\UU_M$ be the set of all qubits that lie in the same partition as their neighbors, \ie, $\UU_M = \lbrace  i \:|\: \NNN(i) \subseteq M \: \vee \: \NNN(i) \subseteq \comp{M} \rbrace$. If there is a qubit $i \in \UU_M$, s.t. $b_i \neq a_i$, then
\label{lem:rule1}
\be
\bra{\vec{b}} \left( \ketbra{\vec{a}}\right)^{T_M} \ket{\vec{b}} = 0
\ee
\end{lemma}
\begin{proof}
Let $g_i, i=1 \dots n$ be the generators defined by Eq.~(\ref{eq:generators}). Since $X^T = X$, $Y^T = -Y$, $Z^T = Z$ and $\eins^T = \eins$, the partial transposition of a product of generators only changes the product's sign. Thus, we can describe the action of the partial transpose $T_M$ w.r.t partition $M$ on products of generators by
\be
\left(\prod \limits_{i=1}^{n} g_i^{x_i}\right)^{T_M} = (-1)^{f(\vec{x})} \left(\prod \limits_{i=1}^{n} g_i^{x_i}\right) \:,
\ee
where $x_i \in \lbrace 0, 1 \rbrace$. Here, $f$ depends on $M$ and is a Boolean function defined by
\begin{align}
f: \lbrace 0,1 \rbrace^{n} &\rightarrow \lbrace 0,1 \rbrace\\
\vec{x} &\mapsto f(\vec{x}) =
\begin{cases}
0, & \mbox{if} \: \left(\prod \limits_{i=1}^{n} g_i^{x_i}\right)^{T_M} = \prod \limits_{i=1}^{n} g_i^{x_i}\\
1, & \mbox{if} \: \left(\prod \limits_{i=1}^{n} g_i^{x_i}\right)^{T_M} = - \prod \limits_{i=1}^{n} g_i^{x_i}
\end{cases}
\nonumber
\:.
\end{align}
Note that the {\it support} ${\rm supp}(f)$ of a Boolean function contains the bits that the function depends on, \ie,
\begin{widetext}
\be
{\rm supp}(f) = \lbrace i\: \vert \: \exists \: \vec{x}, {\rm s.t.} f(x_1, \dots , x_i ,\dots , x_n) \neq f(x_1, \dots , x_i \oplus 1,\dots , x_n)  \rbrace \:.
\ee
Due to the explicit form of the $g_i$, flipping the value of $x_i$ cannot change $f(\vec{x})$, if $i \in \UU_M$. Therefore, $\UU_M \cap \: {\rm supp}(f) = \lbrace \rbrace$. For this reason, we can pull qubits in $\UU_M$ out of the partial transposition $T_M$ in the following way (using also Eq.~(\ref{eq:projectorform})).
\be
\bra{\vec{b}} \left( \ketbra{\vec{a}}\right)^{T_M} \ket{\vec{b}} = \:\trace \left[ \prod \limits_{j = 1}^{n} \frac{(-1)^{b_j} g_j+\eins}{2} \prod \limits_{i \in \UU_M} \frac{(-1)^{a_i} g_i+\eins}{2}\left( \prod \limits_{i \notin \UU_M} \frac{(-1)^{a_i} g_i+\eins}{2}\right)^{T_M}\right]
\ee
\end{widetext}
Since $\frac{g_i + \eins}{2} \frac{- g_i + \eins}{2} = 0$, the last expression vanishes if there is an $i \in \UU_M$, such that $b_i \neq a_i$.
\end{proof}
\begin{lemma}
\label{lem:rule2}
Given a state $\ket{\psi}$ and its Schmidt decomposition $\ket{\psi} = \sum_{i=1}^{d_1} \lambda_i \ket{\mu_i} \otimes \ket{\nu_i}$ with respect to some bipartition $M|\comp{M}$, where $\lambda_i \geq 0$, $d_1 =  {\rm dim}(M)$, $d_2 = {\rm dim}(\comp{M})$ and w.l.o.g. $d_1 \leq d_2$. Note that, here, $\ket{\mu_i} \otimes \ket{\nu_i}$ is not to be understood in a graph state basis. Then, for any state $\ket{\phi}$,
\be
\bra{\phi} \left( \ketbra{\psi} \right)^{T_M} \ket{\phi} \leq \max_{i} \lambda_i^2
\ee
\end{lemma}
\begin{proof}
Writing down $\left( \ketbra{\psi} \right)^{T_M}$ in the basis $\lbrace\ket{\mu_i} \otimes \ket{\nu_i}\rbrace_{i=1 \dots d_{1}, j=1 \dots d_{2}}$, one obtains a matrix with two different kinds of submatrices. First, a diagonal one with diagonal elements $\lambda_i^2$ or zero. Second, anti-diagonal submatrices of the form
\be
\begin{pmatrix}
0 & \lambda_i \lambda_j \\
\lambda_i \lambda_j & 0 \\
\end{pmatrix}
\:.
\ee
Thus, the eigenvalues of the total matrix are $\lbrace \pm \lambda_i \lambda_j, \lambda_i^2, 0 \rbrace$ and the maximum of these eigenvalues has the form $\lambda_i^2$.\\ 
\end{proof}

For the following lemma, we need to keep in mind that the application of the Pauli operator $Z_k$ to a graph state basis vector creates a bit flip on bit $k$, \ie,
\be
\label{eq:flipprop}
Z_k \ket{\vec{a}} = \ket{a_1 \dots a_{k-1} \: a_k \oplus 1 \: a_{k+1} \dots a_n} \: .
\ee
\begin{lemma}
\label{lem:rule3}
Given a graph $G$. Then, in the associated graph state basis,
\begin{align}
\label{eq:rule3}
\left( \ketbra{\vec{a}} + \ketbra{\vec{b}}\right)^{T_k} =  \ketbra{\vec{a}} + \ketbra{\vec{b}}\:,
\end{align}
if $\ket{\vec{b}}$ is obtained from $\ket{\vec{a}}$ in one of the two following ways.
\begin{enumerate}[(i)]
\item $\ket{\vec{b}} = Z_k \ket{\vec{a}}$.
\item $\ket{\vec{b}} = \prod_{i \in \NN(k)}Z_i \ket{\vec{a}}$.
\end{enumerate}
\end{lemma}

\begin{proof}
\begin{enumerate}[(i)]
\item With Eq.~(\ref{eq:projectorform}) and Eq.~(\ref{eq:flipprop}), we have
\begin{align}
\label{eq:nogk}
& \ketbra{\vec{a}} + \ketbra{\vec{b}} \nonumber\\
= & \: \left(\frac{g_k+\eins}{2} + \frac{-g_k+\eins}{2} \right) \prod_{\substack{i = 1\\i \neq k}}^{n} \frac{(-1)^{a_i} g_i+\eins}{2} \nonumber \\
= & \: \prod_{\substack{i = 1\\i \neq k}}^{n} \frac{(-1)^{a_i} g_i+\eins}{2}\:.
\end{align}
Since $g_k$ cancels in Eq.~(\ref{eq:nogk}), the explicit form of the generators $g_i$ implies that, in this equation, there is no $Y$ on qubit $k$. Since $Y$ is the only Pauli matrix that changes under partial transposition, $\ketbra{\vec{a}} + Z_k \ketbra{\vec{a}} Z_k$ is invariant under $T_k$.
\begin{widetext}
\item Using again Eqs.~(\ref{eq:projectorform}) and (\ref{eq:flipprop}) yields
\be
\label{eq:2}
\ketbra{\vec{a}} + \ketbra{\vec{b}}= \left( \prod_{i \in \NN(k)} \frac{(-1)^{a_i} g_i+\eins}{2} + \prod_{i \in \NN(k)} \frac{(-1)^{a_i+1} g_i+\eins}{2} \right) \frac{(-1)^{a_k} g_k+\eins}{2} \prod_{i \notin \NNN(k)} \frac{(-1)^{a_i} g_i+\eins}{2}\:.
\ee
The expression in the first brackets can be simplified to
\begin{align}
\label{eq:simplification}
\prod_{i \in \NN(k)} \frac{(-1)^{a_i} g_i+\eins}{2} + \prod_{i \in \NN(k)} \frac{(-1)^{a_i+1} g_i+\eins}{2} = & \: 2^{-N} \left( \sum_{\vec{x} \in \lbrace 0,1 \rbrace^ {N}} (-1)^{\vec{a} \vec{x}} \prod_{i \in \NN(k)} g_i^{x_i}+ \sum_{\vec{x} \in \lbrace 0,1 \rbrace^{N}} (-1)^{\vec{a} \vec{x} +\sum_i x_i} \prod_{i \in \NN(k)} g_i^{x_i} \right) \nonumber \\
= & \: 2^{-N+1} \left( \sum_{\rm even} (-1)^{\vec{a} \vec{x}} \prod_{i \in \NN(k)} g_i^{x_i}\right)\:,
\end{align}
\noindent where we defined $N = \vert \NN(k) \vert$ and $\sum_{\rm even}$ is the sum over all $\vec{x} \in \lbrace 0,1 \rbrace^N$ for which $\sum_{i=1}^{N} x_i$ is even.\\
Combining Eqs.~(\ref{eq:2}) and (\ref{eq:simplification}) yields
\begin{align}
\label{eq:4}
\ketbra{\vec{a}} + \ketbra{\vec{b}}= &\: 2^{-N} \left( \sum_{\rm even} (-1)^{\vec{a} \vec{x}} \prod_{i \in \NN(k)} g_i^{x_i}\right) (-1)^{a_k} g_k \prod_{i \notin \NNN(k)} \frac{(-1)^{a_i} g_i+\eins}{2} \nonumber \\
& +2^{-N} \left( \sum_{\rm even} (-1)^{\vec{a} \vec{x}} \prod_{i \in \NN(k)} g_i^{x_i}\right) \prod_{i \notin \NNN(k)} \frac{(-1)^{a_i} g_i+\eins}{2} \:.
\end{align}
\end{widetext}
In Eq.~(\ref{eq:4}), the right-hand side is invariant under $T_k$. To see this, consider the first part of the sum, in which all terms contain $g_k$ and an even number of generators in $\NN(k)$. Together with the form of the generators [see Eq.~(\ref{eq:generators})], this implies that there is no $Y$ on qubit $k$.\\
The second part of the right-hand side of Eq.~(\ref{eq:4}) does not have a $Y$ on qubit $k$ either, since it does not contain $g_k$. Consequently, Eq.~(\ref{eq:4}) is invariant under $T_k$.
\end{enumerate}
\end{proof}

Let us now return to the main proof in which we will need Lemmas \ref{lem:rule1}, \ref{lem:rule2} and \ref{lem:rule3}.

\begin{proposition}
\label{lem:witness}
Let $\ketbra{Cl_n}$ be an n-qubit linear cluster state with $n>3$ and $g_i$ the generators of its stabilizer group. Let $\BB = \lbrace \beta_i \rbrace$ be a subset of all qubits such that $\NNN(\beta_i) \cap \NNN(\beta_j) = \lbrace \rbrace$ for $i \neq j$. We define $m = \vert \BB\vert$. Let $\sum_{\vec{s}}$ be the sum over all vectors $\vec{s}$ of length m with elements $s_i = \pm 1$ that contain at least two elements that are $-1$, \ie, $\sum_{i=1}^{m} s_i \leq m -4$. Then,
\begin{align}
\label{eq:witnessform}
W_{\rm Cln} = \frac{1}{2}\eins - \ketbra{Cl_n} - \frac{1}{2} \left( \sum \limits_{\vec{s}} \prod \limits_{i \in \BB}\frac{s_i g_i+\eins}{2} \right) 
\end{align}
is a fully decomposable witness for $\ketbra{Cl_n}$, \ie, for any strict subset $M$ of all qubits, it can be written in the form $W_{Cl} = P_M + Q_M^{T_M}$, where $P_M \geq 0, Q_M \geq 0$.
\end{proposition}
For the sake of brevity, we define $P_{+} = \sum_{\vec{s}} \prod_{i \in \BB}\frac{s_i g_i+\eins}{2}$. Note that $P_{+}$ is a sum of all projectors onto graph state basis vectors that contain at least two excitations in $\BB$, \ie, two bits that equal one. For example, if we choose $\BB = \lbrace 1, 4, 7, \dots \rbrace$ for the sake of illustration,
\begin{align}
\label{eq:6}
P_{+} = \sum_{\vec{x} \in \lbrace 0,1 \rbrace^{n-m}} & \left( \ketbra{0 x_1 x_2 1 x_3 x_4 1 \dots} \right.\nonumber \\
+ \: &\ketbra{1 x_1 x_2 0 x_3 x_4 1 \dots} \nonumber \\
+ \: &\ketbra{1 x_1 x_2 1 x_3 x_4 0 \dots} \nonumber \\
+ \: &\ketbra{1 x_1 x_2 1 x_3 x_4 1 \dots} \nonumber \\
+ \: & \left. \dots \right)
\end{align}
This also illustrates why the presented construction does not work for linear cluster states that consist of three or less qubits. In such a case, $\BB$ could only have one element, since the elements' neighborhoods cannot overlap. Then, however, one cannot fulfill the condition that there must always be two bits in $\BB$ that equal one. Note that sets $\BB = \lbrace 1,4 \rbrace$ and $\BB = \lbrace 1,4,7 \rbrace$ for four and seven qubits, respectively, result in the witnesses $W_{\rm Cl4}$ and $W_{\rm Cl7}$ of Eqs.~(6) and (8) of the main text. Moreover, these choices of $\BB$ lead to the noise tolerance of Eq.~(9) of the main text.
\begin{proof}
First, we will provide a way to construct an appropriate operator $P_M \geq 0$ for every $M$. The main trick is to do this in such a way that this operator is invariant under certain partial transpositions. Second, we will use these $P_M$ to show that $Q_M = \left(W-P_M\right)^{T_M} \geq 0$ for every $M$.\\
Note that we order the qubits $\beta_i$ in a canonical way such that $\beta_i < \beta_{i+1}$. To simplify notation, we define bitstring $\vec{w}_i$ for every $\NNN(\beta_i)$ (of $\vert \NNN(\beta_i) \vert$ bits length) in such a way that its bits define to which partition the correspoding qubits belongs. $\vec{w}_i = (\omega_{-1}, \omega_{0}, \omega_{+1}),\: \omega_j \in \lbrace 0,1 \rbrace,$ means that qubit $\beta_{i}+j$ does not belong to $M$ if $\omega_{j} = 0$ and it belongs to $M$ if $\omega_{j} = 1$. For $\beta_i$ being the first qubit, we consider $\vec{w}_i = (\omega_{0}, \omega_{+1})$; for $\beta_i$ being the last qubit, we use $\vec{w}_i = (\omega_{-1}, \omega_{0})$.\\
Given $M$, we proceed in the following way to construct $P_M$.
\begin{enumerate}
\item Start with $P^{(0)}_M = \ketbra{Cl_n} = \ketbra{0 \dots 0}$.
\item Loop through the values of $i = 1, 2, \dots, m$ and update $P^{(i)}_M$ until you reach $i > m$. Then proceed with step 3. The update of $P^{(i)}_M$ is carried out in the following way.
\begin{enumerate}[(i)]
\item If $\vert \NNN(\beta_i) \vert = 2$:
\begin{enumerate}[I]
\item If $\vec{w}_i = 01$ or $\vec{w}_i = 10$, then $P^{(i)}_M = P^{(i-1)}_M + Z_{k} P^{(i-1)}_M Z_{k}$, where $k \in \NN(\beta_i)$.
\item If $\vec{w}_i = 11$ or $\vec{w}_i = 00$, let $P^{(i)}_M = P^{(i-1)}_M$.
\end{enumerate}
\item If $\vert \NNN(\beta_i) \vert = 3$:
\begin{enumerate}[I]
\item If $\vec{w}_i = 110$ or $\vec{w}_i = 001$, then $P^{(i)}_M = P^{(i-1)}_M + Z_{\beta_i+1} P^{(i-1)}_M Z_{\beta_i+1}$.
\item If $\vec{w}_i = 010$ or $\vec{w}_i = 101$, then $P^{(i)}_M = P^{(i-1)}_M + Z_{\beta_i-1} Z_{\beta_i+1} P^{(i-1)}_M Z_{\beta_i-1} Z_{\beta_i+1}$.
\item If $\vec{w}_i = 100$ or $\vec{w}_i = 011$, then $P^{(i)}_M = P^{(i-1)}_M + Z_{\beta_i-1} P^{(i-1)}_M Z_{\beta_i-1}$.
\item In all other cases, let $P^{(i)}_M = P^{(i-1)}_M$.
\end{enumerate}
\end{enumerate}
\item Let $r$ be the number of times that one of the cases 2.(i).I, 2.(ii).I, 2.(ii).II or 2.(ii).III occurred, \ie, the number of steps in which $P^{(i)}_M$ changed.\\
If $r \leq 1$, define
\be
\label{eq:P_def1}
P_M = 0 \:.
\ee
Let $t$ be the value of $i$ for which $P^{(i)}_M$ was changed the last time, \ie, $P^{(i)}_M = P^{(t)}_M \: \forall \: i> t$. If $r > 1$, define 
\be
\label{eq:P_def}
P_M = P^{(t-1)}_M -\ketbra{Cl_n} \:.
\ee
\end{enumerate}
In other words, $r$ equals the number of bits $\beta_i \in \BB$ which do not lie at the border of $M$ (or $\comp{M}$), i.e. they do not obey 
\be
\label{eq:5}
\NNN(\beta_i) \subseteq M \: {\rm or}  \:\NNN(\beta_i) \subseteq \comp{M}\:.
\ee
Thus, for $r=0$, all $\beta_i \in \BB$ obey this property, for $r =1$, there is one exception.\\
Note that the operator $P_M$ constructed via the given algorithm is either zero or a sum of one-dimensional projectors onto basis states, \ie, 
\be
\label{eq:formofP}
P_M = \sum_{\vec{a}} \ketbra{\vec{a}} \:.
\ee
This can be seen by the fact that $P_M^{(0)} = \ketbra{Cl_n} = \ketbra{0 \dots 0}$, the application of $Z$ only flips a bit and finally $\ketbra{Cl_n}$ is subtracted again. Let us illustrate the algorithm by a concrete example.\\
\textbf{Example of the algorithm:} Consider an eight-qubit cluster state and $M = \lbrace 2,3,5\rbrace$. We choose $\BB = \lbrace 1,4,7 \rbrace$. Then, the algorithm proceeds as follows.\\
\begin{itemize}
\item $P^{(0)}_M = \ketbra{Cl_8} = \ketbra{00000000}$
\item Step $i = 1$, \ie, $\beta_1 = 1$: Since $\vec{w}_1 = 01$, case 2.(i).I applies and 
\begin{align}
P_M^{(1)} = &\: \ketbra{00000000} \nonumber\\
& + Z_2 \ketbra{00000000} Z_2 \nonumber\\
= &\:\ketbra{00000000} \nonumber\\
& + \ketbra{01000000}\:.
\end{align}
\item Step $i = 2$, \ie, $\beta_2 = 4$: Since $\vec{w}_2 = 101$, case 2.(ii).II applies and 
\begin{align}
P_M^{(2)} = &\: P_M^{(1)} + Z_3 Z_5 P_M^{(1)} Z_3 Z_5 \nonumber\\
= & \:\ketbra{00000000} \nonumber\\
&+ \ketbra{01000000} \nonumber\\
&+\ketbra{00101000} \nonumber\\
&+ \ketbra{01101000}\:.
\end{align}
\item Step $i = 3$, \ie, $\beta_3 = 7$: We have $\vec{w}_3 = 000$. Thus, case 2.(ii).IV applies and $P_M^{(3)} = P_M^{(2)}$.
\item As $i = 4 > m = 3$, we abort the loop. $P_M^{(i)}$ changed in two steps, \ie, $r=2$. Moreover, since $P_M^{(i)}$ did not change in the last step, \ie, the third step, $t = 2$. Then,
\begin{align}
P_M = & \:P^{(t-1)}_M -\ketbra{Cl_7} \nonumber\\
= &\:\noindent  P^{(1)}_M -\ketbra{00000000} \nonumber\\
=& \: \ketbra{01000000}
\end{align}
\end{itemize}
Let us now return to the general case and understand the properties of the operator $P_M$ for an arbitrary $M$. The construction uses Lemma \ref{lem:rule3} to ensure that, in every step, either
\begin{subequations}
\label{eq:P_invariance1}
\be
\label{eq:P_invariance1a}
\left(P^{(i)}_M\right)^{T_{M_{i}}} = \left(P^{(i)}_M\right)^{T_{\NNN(\beta_i)}}
\ee
or
\be
\label{eq:P_invariance1b}
\left(P^{(i)}_M\right)^{T_{M_{i}}} = P^{(i)}_M
\ee
\end{subequations}
hold, where we defined $M_k = M \cap \NNN(\beta_k)$.\\
To see that Eqs.~(\ref{eq:P_invariance1}) hold, it is best to consider an example: Assume that $\vec{w}_i = 010$. Then, case 2.(ii).II of the algorithm applies. As seen above, $P_M^{(i-1)}$ has the form $P_M^{(i-1)} = \sum_{\vec{b}} \ketbra{\vec{b}}$. Thus, 
\be
P_M^{(i)} = \sum_{\vec{b}} \ketbra{\vec{b}} +  Z_{i-1} Z_{i+1} \ketbra{\vec{b}} Z_{i+1} Z_{i-1} \: .
\ee
Then, Lemma 2.(ii) implies that Eq.~(\ref{eq:P_invariance1b}) holds.\\
Eqs.~(\ref{eq:P_invariance1}) hold in every step, \ie, for $i = j$ and for $i=k$, where $j \neq k$. According to the premise of non-overlapping neighborhoods of the qubits in $\BB$, we have $\NNN(\beta_j) \cap \NNN(\beta_k) = \lbrace \rbrace$. Therefore, the partial transpositions in Eqs.~(\ref{eq:P_invariance1}) for $i=j$ always affect qubits different from the ones that are affected by the partial transpositions for $i=k$. For this reason, Eqs.~(\ref{eq:P_invariance1}) for $P_M^{(t-1)}$ hold with respect to to every value of $k$, except for $k=t$. More precisely,
\begin{subequations}
\label{eq:P_invariance}
\be
\left(P_M^{(t-1)}\right)^{T_{M_{k}}} = \left(P_M^{(t-1)}\right)^{T_{\NNN(\beta_k)}} 
\ee
or
\be
\left(P_M^{(t-1)}\right)^{T_{M_{k}}} =P_M^{(t-1)} 
\ee
\end{subequations}
is true for every $k \neq t$.
We will use this important property later.\\
Let us proceed with the proof. Since $P_M$ is zero or has the form of Eq.~(\ref{eq:formofP}), we know that $P_M \geq 0$. Thus, it remains to show that $\left(W_{Cl} - P_M\right)^{T_M} = Q_M \geq 0, \; \forall \: M$.\\
Operators that are diagonal in a graph state basis can be written in the form $\sum_{\vec{x}} \prod_{i=1}^{n} g_i^{x_i}$, where the sum runs over binary strings $\vec{x}$ that depend on the operator. Since any partial transposition can at most introduce minus signs in some terms of this sum, such operators remain diagonal under any partial transposition. Therefore, it is enough to prove that 
\be
\label{eq:diagonal_positivity}
\bra{\vec{k}}\left(W_{Cl} - P_M\right)^{T_M}\ket{\vec{k}} \geq 0, \;\forall \: M, \ket{\vec{k}}\:.
\ee
Note that $P_{+} = \sum_{\vec{s}} \prod_{i \in \BB}\frac{s_i g_i+\eins}{2}$ is invariant under any partial transposition, since it does not contain any $Y$ operators due to the form of the generators and the fact that no two qubits in $\BB$ have a common neighbor. Therefore, together with the explicit form of the witness given in Eq.~(\ref{eq:witnessform}), we can rewrite Eq.~(\ref{eq:diagonal_positivity}) as
\be
\label{eq:diagonal_positivity2}
\frac{1}{2} - \frac{1}{2}  \bra{\vec{k}} P_{+} \ket{\vec{k}}- \bra{\vec{k}}\left(\ketbra{Cl_n} + P_M\right)^{T_M}\ket{\vec{k}} \geq 0, \; \forall \: M, \ket{\vec{k}}\:.
\ee
In order to prove this, we distinguish two different cases:
\begin{enumerate}
\item $\boxed{\bra{\vec{k}} P_{+} \ket{\vec{k}} \neq 0 \Leftrightarrow \bra{\vec{k}} P_{+} \ket{\vec{k}} = 1}$\\
\vskip 0.1cm
Note that the equivalence can best be see in Eq.~(\ref{eq:6}). Also, Eq.~(\ref{eq:6}) implies that there must be at least two bits of $\vec{k}$, say $k_{i_0},k_{j_0} \in \BB$, with $i_0 \neq j_0,$ such that $k_{i_0} = k_{j_0} = 1$.\\
In the case $P_M = 0$, Eq.~(\ref{eq:diagonal_positivity2}) and $\bra{\vec{k}} P_{+} \ket{\vec{k}} = 1$ are equivalent to
\be
\label{eq:0}
- \bra{\vec{k}}\left(\ketbra{Cl_n}\right)^{T_M}\ket{\vec{k}}  \geq 0 \:\forall \: M, \ket{\vec{k}}\:.
\ee
To see that the left-hand side always vanishes, one uses that $P_M = 0$ is equivalent to $r \leq 1$, \ie, cases 2.(i).I or 2.(ii).I or 2.(ii).II or 2.(ii).III in the algorithm occurred at most once. This means that $\NNN(\beta_i) \subseteq M$ or $\NNN(\beta_i) \subseteq \comp{M}$ holds for all qubits $\beta_i \in \BB$ with at most one exception, namely $\beta_t$. With $k_{i_0} = k_{j_0} = 1$, Lemma 1 can be applied to see that Eq.~(\ref{eq:0}) vanishes.\\
In the case $P_M \neq 0$, Eq.~(\ref{eq:diagonal_positivity2}) can be simplified using $\bra{\vec{k}} P_{+} \ket{\vec{k}} = 1$ to
\begin{align}
\label{eq:1}
& - \bra{\vec{k}}\left(\ketbra{Cl_n} + P_M\right)^{T_M}\ket{\vec{k}}  \geq 0 \:\forall \: M, \ket{\vec{k}}\nonumber \\
\Leftrightarrow & \: - \bra{\vec{k}}\left(P^{(t-1)}_M\right)^{T_M}\ket{\vec{k}}\geq 0 \:\forall \: M, \ket{\vec{k}}\:.
\end{align}
Here, the definition of $P_M$, Eq.~(\ref{eq:P_def}), has been used.\\
Now, $P_M$ and therefore $P_M^{(t-1)}$ consists of a sum of projectors onto graph basis states $\ket{\vec{a}}$ (see Eq.~(\ref{eq:formofP})). Since the algorithm starts with $P_M^{(0)} = \ketbra{0 \dots 0}$ and never flips any bits on the qubits $\beta_i \in \BB$, these states $\ket{\vec{a}}$ obey $a_{\beta_i} = 0, \; \forall \: i = 1, \dots, m$. Also, depending on whether $i_0 = t$ or $j_0 = t$, Eqs.~(\ref{eq:P_invariance}) can be applied to whichever of these two qubits is different from $t$. The invariance given by Eqs.~(\ref{eq:P_invariance}) then implies that this qubit can be treated as if it lay on the border of $M$ (or $\comp{M}$) and therefore Lemma \ref{lem:rule1} yields
\be
\bra{\vec{k}}\left(P^{(t-1)}_M\right)^{T_M}\ket{\vec{k}} = 0
\ee
and therefore Eq.~(\ref{eq:1}) holds.
\item $\boxed{\bra{\vec{k}} P_{+} \ket{\vec{k}} = 0}$\\
\vskip 0.1cm
To show that Eq.~(\ref{eq:diagonal_positivity2}) holds, we need to prove that 
\be
\label{eq:7}
\bra{\vec{k}}\left(\ketbra{Cl_n} + P_M \right)^{T_M}\ket{\vec{k}} \leq  \frac{1}{2} \: .
\ee
In the case $P_M \neq 0$, $P_M$ is given by Eq.~(\ref{eq:formofP}) and Eq.~(\ref{eq:7}) is equivalent to 
\be
\label{eq:8}
\bra{\vec{k}}\left(\ketbra{Cl_n} + \sum_{\vec{a}} \ketbra{\vec{a}} \right)^{T_M}\ket{\vec{k}} \leq  \frac{1}{2} \: .
\ee
Note that $\ketbra{Cl_n} + \sum_{\vec{a}} \ketbra{\vec{a}} = P_M^{(t-1)}$ consists of $2^{r-1}$ terms, as one starts with one term and doubles this number $(r-1)$ times to obtain $P_M^{(t-1)}$. Therefore, it is enough to show that
\be
\bra{\vec{k}}\left(\ketbra{Cl_n}\right)^{T_M}\ket{\vec{k}} \leq  2^{-r}
\ee
and
\be
\bra{\vec{k}}\left(\ketbra{\vec{a}}\right)^{T_M}\ket{\vec{k}} \leq  2^{-r} \: \forall \: \ket{\vec{a}}
\ee
holds. We will do this by using Lemma \ref{lem:rule2}. However, since the vectors $\ket{\vec{a}}$ are basis vectors of the graph state basis of $\ket{Cl_n}$, $\ket{\vec{a}}$ and $\ket{Cl_n}$ are LU-equivalent. Therefore, they have the same Schmidt coefficients and Lemma \ref{lem:rule2} results in the same upper bounds. For this reason, it suffices to show only one of these upper bounds, namely
\be
\label{eq:9}
\bra{\vec{k}}\left(\ketbra{Cl_n}\right)^{T_M}\ket{\vec{k}} \leq  2^{-r} \:.
\ee
The idea behind the subsequent proof of Eq.~(\ref{eq:9}) is that, every time $P_M^{(i)}$ changes in our algorithm, $r$ is increased, but at the same time, a Bell pair which has Schmidt coefficients $\lbrace \frac{1}{\sqrt{2}},\frac{1}{\sqrt{2}}\rbrace$ is created which results in a smaller upper bound according to Lemma \ref{lem:rule2}.\\
To easily calculate the Schmidt coefficients of $\ket{Cl_n}$ with respect to bipartition $M \vert \comp{M}$, we transform $\ket{Cl_n}$ via local transformations into disconnected Bell pairs whose qubits are in different partitions. In order not to change the Schmidt coefficients during this transformation, we apply the following two operations to the graph.
\begin{figure}
\includegraphics[width=0.9\columnwidth]{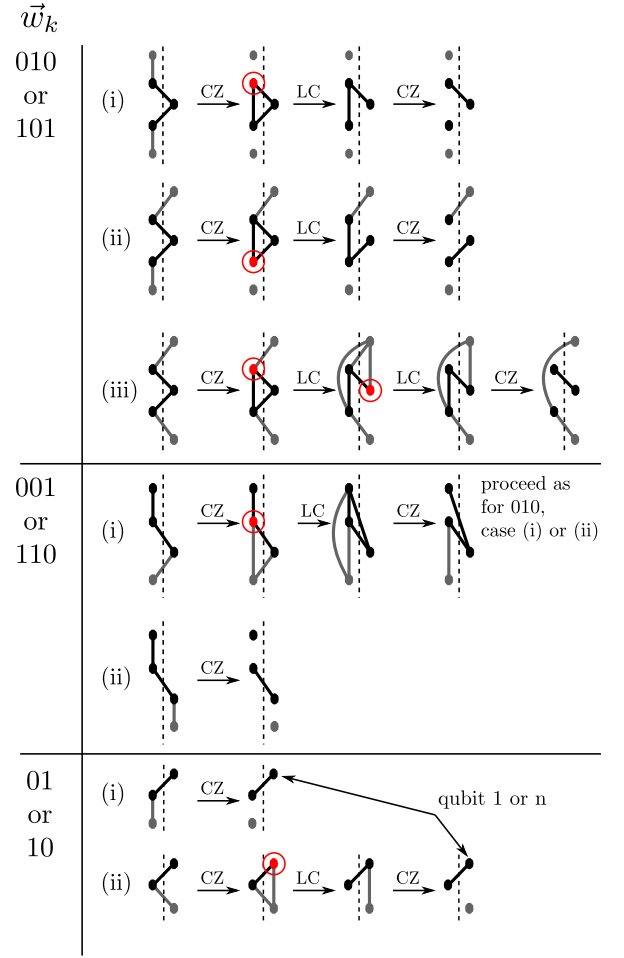}
\caption{\label{fig:threequbittrafos} For every case in which $P^{(i)}_M$ changes in the algorithm, we explicitly demonstrate that it is possible to create a Bell pair via a sequence of local transformations. Note that the cases $\vec{w}_k = 100$ and $\vec{w}_k = 011$ follow from $\vec{w}_k = 110$ or $\vec{w}_k = 001$ due to symmetry reasons.}
\end{figure}
\begin{itemize}
\item {\it Local Complementation (LC)}: Pick a qubit $j$ and invert its neighborhood graph. In other words, if between two neighbors of $j$ there is an edge, delete it. If between two neighbors of $j$ there is no edge, add an edge between them. Since this operation corresponds to the application of local unitaries, the Schmidt coefficients remain unchanged.
\item {\it Controlled-Z (CZ)}: Apply a controlled-Z operation to two qubits in the same partition. If they are connected by an edge, this will delete the edge. If they are not connected, this will create an edge. As this unitary operation is only applied on qubits in the same partition, it does not change the Schmidt coefficients.
\end{itemize}
Since these transformations can be applied to each set $\NNN(\beta_j)$ individually to create a Bell pair between the two partitions (see Fig. \ref{fig:threequbittrafos}) it is possible to end up with distinct Bell pairs between the two partitions.\\
Moreover, the number $q$ of Bell pairs that one can create by transforming the qubits in each $\NNN(\beta_j)$ separately obeys $q \geq  r$, as every time that $P_M^{(i)}$ is changed in the algorithm (including step $i=t$) corresponds to a configuration shown in Fig. \ref{fig:threequbittrafos}. Furthermore, each of these configurations allow for the creation of a disconnected Bell pair.\\
Since the Schmidt coefficients of a single Bell pair are $\lbrace \frac{1}{\sqrt{2}}, \frac{1}{\sqrt{2}} \rbrace$, Lemma~\ref{lem:rule2} yields
\be
\bra{\vec{k}}\left(\ketbra{Cl_n}\right)^{T_M}\ket{\vec{k}} \leq  2^{-q} \:,
\ee
With $q \geq r$, this implies that Eq.~(\ref{eq:9}) holds.\\
In the case $P_M = 0$, we need to show that 
\be
\label{eq:3}
\bra{\vec{k}}\left(\ketbra{Cl_n}\right)^{T_M}\ket{\vec{k}} \leq \frac{1}{2} \:.
\ee
Note that $P_M = 0$ is equivalent to $r \leq 1$. On one hand, $r = 0$ means that for each $i = 1, \dots, n$, either $\NNN(i) \subseteq M$ or $\NNN(i) \subseteq \comp{M}$ holds. Since $M \neq \lbrace 1, \dots, n\rbrace$, deleting all edges within $M$ and $\comp{M}$ will lead to at least one Bell pair, which has Schmidt coefficients $\lbrace \frac{1}{\sqrt{2}}, \frac{1}{\sqrt{2}} \rbrace$. With Lemma~\ref{lem:rule2}, Eq.~(\ref{eq:3}) holds.\\
On the other hand, for $r=1$, Eq.~(\ref{eq:3}) also holds, since due to Fig. \ref{fig:threequbittrafos} the Schmidt coefficients are given by at least one Bell pair. Therefore, the maximal coefficient is $\frac{1}{2}$.\\
This finishes the proof of Proposition \ref{lem:witness}.
\end{enumerate}
\end{proof}

\bibliographystyle{apsrev}

\begin{thebibliography}{15}
\bibitem{PPTdef}
A state $\vr=\sum_{ijkl} \vr_{ij,kl} \ket{i}\bra{j}\otimes \ket{k}\bra{l}$ is PPT, if its partial transpose
$\vr^{T_A}=\sum_{ijkl} \vr_{ji,kl} \ket{i}\bra{j}\otimes \ket{k}\bra{l}$ has no negative eigenvalues.
\bibitem{states}
The considered states are given by $\ket{GHZ_3} = (\ket{000} + \ket{111})/\sqrt{2}$, $\ket{W_3} = (\ket{001}+\ket{010}+\ket{100})/\sqrt{3}$, $\ket{W_4} = (\ket{0001}+\ket{0010}+\ket{0100}+\ket{1000})/2$, $\ket{Cl_4} = (\ket{0000}+\ket{0011} + \ket{1100} -\ket{1111})/2$, $\ket{D_{2,4}} = (\ket{0011}+\ket{1100}+\ket{0101}+\ket{0110}+\ket{1001}+\ket{1010})/\sqrt{6}$ and $\ket{\Psi_{S,4}} = (\ket{0011}+\ket{1100}-\frac{1}{2}(\ket{0101}+\ket{0110}+\ket{1001}+\ket{1010}))/\sqrt{3}$.
\end{thebibliography}
\begin{thebibliography}{15}
\bibitem{subscripts}
The subscripts $M$ do not matter here.
\end{thebibliography}

\end{document}